\documentclass[12pt,reqno]{amsart}
\usepackage{amsfonts,amsmath,amscd,amssymb,amsthm}
\usepackage{color,enumerate,graphicx,hyperref,perpage,url}
\usepackage[margin=2.5cm]{geometry}

\theoremstyle{plain}
\newtheorem{thm}{Theorem}[section]

\newtheorem{defnn}[thm]{Definition}
\newtheorem{lemma}[thm]{Lemma}

\newtheorem{prop}[thm]{Proposition}

\theoremstyle{remark}
\newtheorem*{defn}{\textbf{Definition}}
\newtheorem*{ex}{\textbf{Example}}
\newtheorem*{rmk}{\textbf{Remark}}

\numberwithin{equation}{section}

\MakePerPage{footnote} \allowdisplaybreaks 

\newcommand{\HH}{\mathcal H}
\newcommand{\Id}{\mathrm{Id}}

\newcommand{\M}{\mathbb M}
\newcommand{\N}{\mathbb N}
\newcommand{\ol}{\overline}
\newcommand{\Op}{\mathrm{Op}}

\newcommand{\R}{\mathbb{R}}

\newcommand{\SC}{\mathcal{S}}
\newcommand{\SL}{\mathrm{SL}}

\newcommand{\T}{\mathbb{T}}
\newcommand{\Tr}{\mathrm{Tr}\,}

\newcommand{\ve}{\varepsilon}

\newcommand{\Vol}{\mathrm{Vol}}

\newcommand{\wt}{\widetilde}
\newcommand{\Z}{\mathbb{Z}}

\title{Small scale quantum ergodicity in cat maps. I}

\author{Xiaolong Han}
\email{xiaolong.han@csun.edu}
\address{Department of Mathematics, California State University, Northridge, CA 91330, USA}

\thanks{} 
\date{}

\begin{document}
\maketitle

\begin{abstract}
In this series, we investigate quantum ergodicity at small scales for linear hyperbolic maps of the torus (``cat maps''). In Part I of the series, we prove quantum ergodicity at various scales. Let $N=1/h$, in which $h$ is the Planck constant. First, for all integers $N\in\N$, we show quantum ergodicity at logarithmical scales $|\log h|^{-\alpha}$ for some $\alpha>0$. Second, we show quantum ergodicity at polynomial scales $h^\alpha$ for some $\alpha>0$, in two special cases: $N\in S(\N)$ of a full density subset $S(\N)$ of integers and Hecke eigenbasis for all integers.
\end{abstract}

\section{Introduction}
One of the main problems in \textit{Quantum Chaos} is to study the distribution of eigenstates in the quantized system for which the classical dynamics is chaotic (i.e. hyperbolic). In this series, we consider the classical dynamics given by a hyperbolic linear map of the torus $\T^2=\R^2/\Z^2$, commonly referred as a (classical) ``cat map'' due to Arnold \cite{AA}. Such a map is defined by a matrix $M\in\SL(2,\Z)$ with $|\Tr M|>2$. Its iterations $M^t$, $t\in\Z$, induce a discrete dynamical system that is hyperbolic \cite{KH}. 

The quantized system of a classical cat map, i.e. a quantum cat map, was proposed by Hannay-Berry \cite{HB}. In brief, this procedure restricts the Planck constant $h$ to be the inverse of an integer: $h=1/N$ for $N\in\N$. The quantum cat map $\hat M$ acts on the $N$-$\dim$ Hilbert space $\HH_N$ of quantum states. There is an eigenbasis (orthonormal basis of eigenstates) $\{\phi_j\}_{j=1}^N$ of $\hat M$ in $\HH_N$. See Section \ref{sec:cat} for the details of quantization in cat maps. The Quantum Ergodicity (QE) theorem \cite{Sn, Ze1, CdV} in the context of cat maps is proved by \cite{BDB, Ze2}. It asserts that a full density (see \eqref{eq:density}) eigenstates equidistribute. 

In this series, we investigate equidistribution of the eigenstates in quantum cat maps on balls $B(x,r)\subset\T^2$ at small scales $r=r(N)\to0$ as $N\to\infty$, i.e. \textit{small scale quantum ergodicity}. In Part I of the series, we first prove quantum ergodicity in cat maps at logarithmical scales $r=(\log N)^{-\alpha}$ for some $\alpha>0$ (Theorem \ref{thm:SSElog} and \ref{thm:SSQElog}). Second, in two special cases for cat maps, we prove quantum ergodicity at much finer polynomial scales $N^{-\alpha}$ for some $\alpha>0$ (Theorems \ref{thm:SSQEpoly} and \ref{thm:SSQEHecke}).

We point out that there are several quantization procedures of cat maps. See Zelditch \cite{Ze2} for the discussion of these quantization approaches. In this series, we follow the approach initiated by Hannay-Berry \cite{HB} and further studied in \cite{K, DE, DEGI, BDB, KR1, KR2, FNDB} among the large literature in mathematics and physics. In this quantization approach, $\T^2$ is the phase space. Quantum states can then be described as distributions on $\R^1$ that satisfy perodic conditions in both position and momentum variables. We call such decription the \textit{position representation} of the quantum states. Furthermore, due to the nature of cat maps $M\in\SL(2,\Z)$ on $\T^2$, there is a rich arithmetic structure which can be used to study equidistribution of the eigenstates. It is explored by Degli Esposti-Graffi-Isola \cite{DEGI}, Kurlberg-Rudnick \cite{KR1, KR2}, Bourgain \cite{Bo}, etc. Quantum ergodicity at polynomial scales (Theorems \ref{thm:SSQEpoly} and \ref{thm:SSQEHecke}) in this paper applies these results directly.

The quantization precedure described above is rather restrictive to the fact that the phase space for the quantum cat maps is $\T^2$. In fact, quantization can be done for a much greater class of maps on general manifolds. That is, cat maps $M\in\SL(2,\Z)$ on $\T^2$ are the simplest examples of symplectic maps on compact symplectic manifolds (as the phase spaces). Equipped with a natural complex structure, the phase space can be regarded as a K\"ahler manifold. The quantum system can then be induced as the restriction of the classical system to the holomorphic sections on the K\"ahler manifold. See e.g. Berezin \cite{Bere} for a long history in this framework of quantization. 

For a large class of symplectic maps, Zelditch \cite{Ze2} defined their quantization as Toeplitz operators. In such quantization approach, the quantum states are holomorphic sections on the K\"ahler manifold. Under the same setting \cite{Ze2}, Chang-Zelditch \cite{CZ} recently established quantum ergodicity at logarithmical scales for symplectic maps which satisfy appropriate conditions, including the cat maps. The holomorphic representation of quantum states \cite{Ze2, CZ} is related to the position representation via Bargmann transform (see \cite[Section 13]{Zw}). We stress that our arguments and results are specifically adapted for cat maps and do not apply to the more general symplectic maps \cite{Ze2, CZ}. Since our discussion of quantum ergodicity is restricted to cat maps, we use the position representation of the quantum states only.

In negatively curved manifolds, the classical dynamics given by the geodesic flow is known to be hyperbolic \cite{KH}. Equidistribution of the eigenstates at similar logarithmical scales (in the physical space, see \eqref{eq:SSE}) was proved by the author \cite{Han1} and Hezari-Rivi\`ere \cite{HR1}. 

However, unlike the case in negatively curved manifolds \cite{Han1, HR1}, in Part II of the series, we show that the logarithmical scales for quantum ergodicity are optimal in cat maps. 

The optimality of logarithmical scales is related to the phenomenon of short periods of the linear maps $M\in\SL(2,\Z)$. That is, define the period $P(N)$ as the smallest positive integer such that $M^{P(N)}=\Id\mod N$. Then we have that $\hat M^{P(N)}=\Id$ in $\HH_N$. It is known that $P(N)\ge\sqrt N$ for almost all $N\in\N$. (See \cite{KR2}.) But there is a sequence $\{N_k\}_{k=1}^\infty\subset\N$ such that 
\begin{equation}\label{eq:logP}
P(N_k)\sim\frac{2\log N_k}{\lambda}\quad\text{as }k\to\infty,
\end{equation}
for which we say $\hat M$ has short periods in $\HH_{N_k}$. Restricting to $\HH_{N_k}$, quantum ergodicity is not valid beyond logarithmical scales. 

The phenomenon of short periods in cat maps also accounts for the scarring (i.e. non-equidistribution, see also the discussion below on Quantum Unique Ergodicity) of some eigenstates proved by Faure-Nonnenmacher-De Bi\`evre \cite{FNDB}, optimal logarithmical rate of quantum ergodicity proved by Schubert \cite{Sc}, and optimality proved by Faure-Nonnenmacher \cite{FN} of the entropy bounds of semiclassical measures \cite{A, AN}. In fact, the proof the optimality of logarithmical scales in quantum ergodicity applies the techniques in \cite{FNDB, FN}. See Part II of the series for details. We remark that such phenomenon is not shared by some other hyperbolic systems, in particular, the geodesic flows in negatively curved manifolds, for which the logarithmical scales in quantum ergodicity are unlikely to be optimal.

In Part I, the second main result (Theorems \ref{thm:SSQEpoly}) says that restricting to $N\in\N$ with long periods (in particular, $P(N)\gg\sqrt Ne^{(\log N)^\delta}$ for some $\delta>0$, see Theorem \ref{thm:V4TnNHecke}), quantum ergodicity holds at some polynomial scales $r=r(N)=N^{-\alpha}$, $\alpha>0$. Because $\hat M$ has long periods in $\HH_N$ for almost all integers $N$ (\cite{KR2}), quantum ergodicity holds at polynomial scales for almost all $N\in\N$. Similar argument concludes quantum ergodicity at finer polynomial scales for Hecke eigenbasis (Theorem \ref{thm:SSQEHecke}), but for all integers.

At this point, the background and setup of small scale quantum ergodicity are in order. The most studied classical dynamical system is given by the geodesic flow on a compact Riemannian manifold $\M$. Let $T^*\M=\{x=(q,p):q\in\M,p\in T_q^*\M\}$ be the cotangent bundle of $\M$. Then the geodesic flow is defined as the Hamiltonian flow (with Hamiltonian $H(q,p)=|p|^2_q$) in the phase space as $T^*\M$. The eigenstates in the quantized system are the eigenfunctions of the (positive) Laplacian $\Delta$ on $\M$. 

If the geodesic flow is ergodic (a weaker condition than hyperbolicity), then the Quantum Ergodicity theorem \cite{Sn, Ze1, CdV} asserts that a full density eigenstates in any eigenbasis are equidistributed in the normalized phase space (that is, the cosphere bundle $S^*\M=\{(q,p)\in T^*\M:|p|_q=1\}$.) More precisely, we associate a classical observable $a\in C^\infty_0(T^*\M)$ with a quantized operator $\Op(a)$ acting on the $L^2(\M)$, the space of quantum states. Then for any eigenbasis $\{\phi_j\}_{j=1}^\infty$, $\Delta\phi_j=\lambda_j^2\phi_j$, there is a full density subset $S(\N)$ of integers such that for $j\to\infty$ in $S(\N)$,
\begin{equation}\label{eq:QE}
\langle\Op(a)\phi_j,\phi_j\rangle\to\int_{S^*\M}a\,d\mu\quad\text{for all }a\in C^\infty_0(T^*\M),
\end{equation}
in which $\mu$ is the normalized Liouville measure on $S^*\M$. Here, $S(\N)$ has full density in the integers $\N$ if
$$\lim_{N\to\infty}\frac{\#\{j\in S(\N):1\le j\le N\}}{N}=1.$$
Similarly, we say a subset $S(N)$ of $\{1,...,N\}$ has full density as $N\to\infty$ if
\begin{equation}\label{eq:density}
\lim_{N\to\infty}\frac{\#\{j\in S(N)\}}{N}=1.
\end{equation}
See \cite{Ze3, Sa} for the recent developments in Quantum Ergodicity. 

In small scale quantum ergodicity, we study \eqref{eq:QE} when the classical observable $a$ has support in balls $B(x,r)\subset T^*\M$ with $r=r(\lambda_j)\to0$ as $j\to\infty$. In particular, if $a=\chi_{B(q,r)}$ as the indicator function of a geodesic ball $B(q,r)$ in the physical space $\M$, then
$$\langle\Op(\chi_{B(q,r)})\phi_j,\phi_j\rangle=\int_{B(q,r)}|\phi_j|^2\,d\Vol.$$
Here, $d\Vol$ is the Riemannian volume on $\M$. Therefore, \eqref{eq:QE} is reduced to that (c.f. \cite[Question 1.3]{Han1}) for some full density $S(\N)$ in $\N$,
\begin{equation}\label{eq:SSE}
\int_{B\left(q,r(\lambda_j)\right)}|\phi_j|^2\,d\Vol=\frac{\Vol(B(q,r(\lambda_j)))}{\Vol(\M)}+o\left(r(\lambda_j)^d\right)\quad\text{as }j\to\infty\text{ in }S(\N)
\end{equation}
for all $q\in\M$. Here, $d=\dim\M$. If \eqref{eq:SSE} holds, then we say that the eigenfunctions $\{\phi_j\}$ for $j\in S(\N)$ are equidistributed at scale $r=r(\lambda)$ in the physical space $\M$. It should be distinguished with quantum ergodicity at scale $r=r(\lambda)$, i.e. equidistribution in the phase space $S^*\M$. See Definition \ref{defn:SSQE} for the precise statement for quantum ergodicity at small scales in quantum cat maps.

In negatively curved manifolds, the geodesic flow is hyperbolic \cite{KH}. Berry's random wave conjectures \cite{Berr} suggest that the eigenfunctions of eigenvalue $\lambda^2$ behave like random waves with frequency $\lambda$. Recent results about equidistribution at various polynomial scales of random waves on manifolds were proved in \cite{Han2, HT, CI}. In comparison, we see that the logarithmical scales in \cite{Han1, HR1} are at much weaker scales.\footnote{In fact, it was only shown in \cite{Han1, HR1} that the two sides of \eqref{eq:SSE} are comparable uniformly for $q\in\M$. This is a weaker statement than the uniform equidistribution at small scales in \eqref{eq:SSE}.} In the special case of modular surfaces and restricted to Hecke eigenbasis, \eqref{eq:SSE} at smaller scales $r=\lambda^{-\alpha}$ for some $\alpha>0$ was known by \cite{LS, Y}.

Equidistribution at small scales in the physical space is not only a feature of ergodicity of the geodesic flow. For example, the geodesic flow on the $d$-$\dim$ torus $\T^d$ is integrable (so is not ergodic).\footnote{We remark that the cat map is not the time $1$ map of some Hamiltonian flow on the torus, therefore differs with the geodesic flow on the torus \cite{HR2, LR, GW} signicantly.} But for any eigenbasis in $\T^d$, there is a full density subsequence that is equidistributed at $r=\lambda^{-1/(d-1)+o(1)}$ by \cite{HR2, LR, GW}. In $\T^2$, the scale approaches the Planck scale $1/\lambda$ and in $\T^d$, $d\ge3$, Bourgain \cite{LR} showed that the scale is optimal.

We now consider the classical cat map given by a matrix $M\in\SL(2,\Z):\T^2\to\T^2$ with $|\Tr M|>2$. In this case, the torus $\T^2=\{(q,p):q,p\in\T^1\}$ is the phase space, in which $q$ and $p$ denote the position and momentum variables, respectively. We regard $\T^1$ as the physical space of the position variable $q$. Denote $\hat M$ the quantization of $M$. So $\hat M:\HH_N\to\HH_N$ for $N\in\N$. (See Section \ref{sec:cat} for more background of cat maps.) Let $B_d(x,r)$ be a geodesic ball in $\T^d$ with radius $r$ and center $x$. The first main theorem asserts equidistribution at logarithmical scales in the physical space.

\begin{thm}[Equidistribution at logarithmical scales]\label{thm:SSElog}
For $0\le\alpha<1/2$ and all $N\in\N$, any eigenbasis $\{\phi_j\}_{j=1}^N$ of a quantum cat map $\hat M$ in $\HH_N$ contains a full density subset that equidistributes at scale $r=(\log N)^{-\alpha}$ in the physical space. That is, there is a full density subset $S(N)$ of $\{1,...,N\}$ such that for $j\in S(N)$, 
\begin{equation}\label{eq:SSElog}
\int_{B_1(q,r)}|\phi_j|^2\,d\Vol=\Vol(B_1(q,r))+o(r)\quad\text{as }N\to\infty,
\end{equation}
uniformly for all $q\in\T^1$. 
\end{thm}

To define small scale quantum ergodicity, i.e. equidistribution of eigenstates at small scales in the phase space $\T^2$, in \eqref{eq:QE}, we can no longer choose the indicator function (not smooth) $a=\chi_{B_2(x,r)}$ for $x\in\T^2$ and $r=r(N)\to0$ as $N\to\infty$. To remedy this, we use trigonometric polynomials $b^\pm_{x,r}$ to approximate $\chi_{B_2(x,r)}$ uniformly for all $x\in\T^2$, i.e. 
$$b^-_{x,r}\le\chi_{B_2(x,r)}\le b^+_{x,r}\text{ and }\int_{\T^2}b^\pm_{x,r}\,d\Vol=\Vol(B_2(x,r))+o\left(r^2\right)\quad\text{uniformly for all }x\in\T^2.$$ 
Such choices of approximation by trigonometric polynomials appear naturally in quantum cat maps. See Lemma \ref{lemma:majmin} and \eqref{eq:batxestimates} for their precise properties. With this understanding, we define small scale quantum ergodicity in quantum cat maps.

\begin{defnn}[Small scale quantum ergodicity in quantum cat maps]\label{defn:SSQE}
Let $\hat M$ be a quantum cat map and $G\subset\N$. We say quantum ergodicity at scale $r=r(N)$ holds for $N\in G$ if for any eigenbasis $\{\phi_j\}_{j=1}^N$ of $\hat M$ in $\HH_N$, there is a full density subset $S(N)$ of $\{1,...,N\}$ such that for $j\in S(N)$,
\begin{equation}\label{eq:SSQE}
\left\langle\Op_N(b^\pm_{x,r})\phi_j,\phi_j\right\rangle=\Vol(B_2(x,r))+o(r^2)\quad\text{as }N\to\infty\text{ in }G,
\end{equation}
uniformly for all $x\in\T^2$. 
\end{defnn}

Then we have that
\begin{thm}[Quantum ergodicity at logarithmical scales]\label{thm:SSQElog}
For $0\le\alpha<1/4$ and all $N\in\N$, quantum ergodicity at scale $r=(\log N)^{-\alpha}$ holds. 
\end{thm}

The second main result in this paper treats two cases for which the logarithmical scales in Theorems \ref{thm:SSElog} and \ref{thm:SSQElog} can be significantly improved to polynomials scales. In these two cases, we apply the results of Kurlberg-Rudnick \cite{KR1, KR2}, where the problem of Quantum Unique Ergodicity (QUE) in cat maps is studied. 
 
If QUE holds, then equidistribution in \eqref{eq:QE} is valid for the \textit{whole} sequence of any eigenbasis. While the QUE conjecture in negatively curved manifolds proposed by Rudnick-Sarnak \cite{RS} is still open, some (positive and negative) results are known in different dynamical systems. Hassell \cite{Has} showed that in generic Bunimovich stadia, QUE does not hold. In arithmatic hyperbolic surfaces and restricting to Hecke eigenbasis, QUE has been verified \cite{L, SV, HS, BL}. The Hecke eigenbasis is the joint eigenbasis of a family of commutative group of operators including the Laplacian. (Note that Brooks-Lindenstauss \cite{BL} proved QUE for the joint eigenbasis of the Laplacian and \textit{one} Hecke operator.)
 
In the context of cat maps, if QUE holds for a subset $G\subset\N$, then 
$$\langle\Op_N(a)\phi_j,\phi_j\rangle\to\int_{\T^2}a\,d\mu\quad\text{as }N\to\infty\text{ in }G,$$
for any eigenbasis $\{\phi_j\}_{j=1}^\infty$ in $\HH_N$. Faure-Nonnenmacher-De Bi\`evre \cite{FNDB} proved that QUE does not hold for $G=\N$. That is, along the sequence $\{N_k\}_{k=1}^\infty$ with short periods \eqref{eq:logP}, there are eigenstates in $\HH_{N_k}$ that fail to be equidistributed, which are called ``scarred'' eigenstates. 

On the positive side of QUE in cat maps, Kurlberg-Rudnick \cite{KR2} proved that there is a full density subset $G\subset\N$ such that QUE holds. (Earlier QUE result for a sparse subset of $\N$ was proved in \cite{DEGI}, assuming the Generalized Riemann Hypothesis.) The cat map in $\HH_N$ for $N\in G$ has sufficiently long periods and our quantum ergodicity at polynomial scales is also in this case:
\begin{thm}[Quantum ergodicity at polynomial scales]\label{thm:SSQEpoly}
There is a full density subset $S(\N)$ of integers such that for $N\in S(\N)$,
\begin{itemize}
\item any eigenbasis $\{\phi_j\}_{j=1}^N$ of the quantum cat map $\hat M$ in $\HH_N$ contains a full density subset that equidistributes at scale $r=N^{-\alpha}$, $0\le\alpha<1/12$, in the physical space,
\item quantum ergodicity at scale $r=N^{-\alpha}$, $0\le\alpha<1/16$, holds.
\end{itemize}
\end{thm}

With suitable additional assumptions of the cat map $M$, Kurlberg-Rudnick \cite{KR1} introduced the Hecke theory associated with the cat map $M$. It is the analogue of the Hecke theory in arithmetic hyperbolic surfaces. That is, in $\HH_N$ for each integer $N\in\N$, they define a family of commutative group of unitary operators including the quantum cat map $\hat M$. Then there is a joint eigenbasis for all these operators, similarly called the Hecke eigenbasis. Kurlberg-Rudnick \cite{KR1} then proved QUE for the Hecke eigenbasis for all integers. Using \cite{KR1}, our quantum ergodicity follows at a better polynomial scales than the ones in Theorem \ref{thm:SSQEpoly}.

\begin{thm}[Quantum ergodicity at polynomial scales for Hecke eigenbasis]\label{thm:SSQEHecke}
For all $N\in\N$ and Hecke eigenbasis $\{\phi_j\}_{j=1}^N$ in $\HH_N$, 
\begin{itemize}
\item $\{\phi_j\}_{j=1}^N$ contains a full density subset that equidistributes at scale $r=N^{-\alpha}$, $0\le\alpha<1/10$, in the physical space,
\item quantum ergodicity for $\{\phi_j\}_{j=1}^N$ at scale $r=N^{-\alpha}$, $0\le\alpha<1/12$, holds.
\end{itemize}
\end{thm}

\begin{rmk}
We shall remark the crucial difference of the Hecke theory in cat maps and the one in arithmetic hyperbolic surfaces. In arithmetic hyperbolic surfaces, all eigenbases are \textit{conjectured} to be Hecke eigenbasis \cite{Sa}, which means that QUE for Hecke eigenbasis should imply QUE. 

However, in cat maps, not all eigenbases are Hecke eigenbasis. The variety of eigenbases display very different distribution properties. That is, the Hecke eigenbasis satisfies QUE \cite{KR1} but some other eigenbasis fails QUE \cite{FNDB}. Similarly in the small scale quantum ergodicity, the Hecke eigenbasis satisfies quantum ergodicity at polynomial scales (Theorem \ref{thm:SSQEHecke}), but some other eigenbasis can only equidistribute up to the logarithmical scale (see Part II).

We shall also remark that the polynomial scales in Theorems \ref{thm:SSQEpoly} and \ref{thm:SSQEHecke} are unlikely  optimal.
\end{rmk}


\subsection*{Organization of the paper}
We organize this paper as follows. In Sections \ref{sec:cat}, we review classical and quantum cat maps. In Section \ref{sec:SS}, we gather some results  that are used to prove equidistribution in the physical space and quantum ergodicity at small scales. In Section \ref{sec:SSlog}, we prove equidistribution in the physical space and quantum ergodicity at logarithmical scales, i.e. Theorems \ref{thm:SSElog} and \ref{thm:SSQElog}. In Section \ref{sec:SSpoly}, we prove equidistribution in the physical space and quantum ergodicity at polynomial scales, i.e. Theorems \ref{thm:SSQEpoly} and \ref{thm:SSQEHecke}.

\section{Classical dynamics and quantum dynamics in cat maps}\label{sec:cat}
In this section, we review the background on classical and quantum cat maps. See \cite{HB, BDB, KR1, FNDB} for more details. Here, we mainly follow \cite[Section 6]{BDB}.

\subsection{Classical cat maps}
Consider the quadratic Hamiltonian on the plane $\R^2$
\begin{equation}\label{eq:Ham}
H(q,p)=\frac12\alpha q^2+\frac12\beta p^2+\gamma qp.
\end{equation}
It generates the Hamiltonian flow 
$$M(t):x(0)=(q(0),p(0))\to x(t)=(q(t),p(t))$$
such that
$$\frac{dq(t)}{dt}=\frac{\partial H}{dp}=\beta p+\gamma q\quad\text{and}\quad\frac{dp(t)}{dt}=-\frac{\partial H}{dq}=-\alpha q-\gamma p.$$
So explicitly
$$M(t)=\exp\left\{t\begin{pmatrix}
\gamma & \beta\\
-\alpha & -\gamma
\end{pmatrix}\right\}.$$
If $\gamma^2>\alpha\beta$, then the flow $M(t)$ is hyperbolic with Lyapunov exponent $\lambda=\sqrt{\gamma^2-\alpha\beta}$. Denote
$$M:=M(1)=\exp\left\{\begin{pmatrix}
\gamma & \beta\\
-\alpha & -\gamma
\end{pmatrix}\right\}=\begin{pmatrix}
A & B\\
C & D
\end{pmatrix}.$$
Then $M\in\SL(2,\R):\R^2\to\R^2$ is a hyperbolic map with eigenvalues $\pm e^\lambda$. Notice that throughout the paper, we use $M$ to denote both the hyperbolic map and the matrix that defines it.

\begin{rmk}
We remark that $M\in\SL(2,\R)$ preserves the Liouville measure $d\mu=dqdp$ on $\R^2$. Moreover, define the symplectic product on $\R^2$
\begin{equation}\label{eq:sympprod}
u\wedge v=u_2v_1-u_1v_2\quad\text{for }u=(u_1,u_2),v=(v_1,v_2)\in\R^2.
\end{equation}
Then
\begin{equation}\label{eq:Mpreswedge}
uM\wedge vM=u\wedge v.
\end{equation}
That is, $M$ preserves the symplectic product.
\end{rmk}

\begin{defn}[Classical cat maps]
Let $M\in\SL(2,\R):\R^2\to\R^2$ be a hyperbolic map. Suppose further that $M\in\SL(2,\Z)$, i.e. $A,B,C,D\in\Z$. Since
$$(x+n)M=xM+nM=xM\mod1\quad\text{for }x\in\R^2\text{ and }n\in\Z^2,$$
$M$ induces a map on $\T^2$ that is hyperbolic, by which we refer as a classical cat map. 
\end{defn}

\begin{ex}[Arnold cat map]
The Arnold cat map is defined by
$$M=\begin{pmatrix}
2 & 1\\
1 & 1
\end{pmatrix}.\quad\text{The Lyapunov exponent }\lambda=\frac{3+\sqrt5}{2}.$$
\end{ex}

As mentioned in the introduction, our main theorems of quantum ergodicity in quantum cat maps are closely related to the periods of classical cat maps.
\begin{defn}[Periods of cat maps]
Let $M\in\SL(2,\Z)$ be a classical cat map. Define $P(N)$ as the period (or order) of $M$ module $N$, that is, the smallest positive integer $k\ge1$ for which $M^k=\Id\mod N$.
\end{defn}

The following proposition provides estimates of periods for cat maps \cite{K, KR1, KR2}.
\begin{prop}[Estimates of the periods of cat maps]
Let $M$ be a classical cat map. 
\begin{enumerate}[(i).]
\item There is $C>0$ depending only on $M$ such that
$$\frac2\lambda\log N-C\le P(N)\le3N\quad\text{for all }N\in\N.$$
\item There is a full density subset $S(\N)$ of integers such that
$$P(N)\ge\sqrt N\quad\text{for all }N\in S(\N).$$
\item There is a sequence of integers $\{N_k\}_{k=1}^\infty$ such that
$$P(N_k)=\frac{2\log N_k}{\lambda}+O(1)\quad\text{as }k\to\infty.$$
\end{enumerate}
\end{prop}

\subsection{Quantum cat maps}
We first recall the quantization on the real line $\R$, in which case the phase space is $\R^2$. For a detailed discussion, see \cite[Chapter 4]{Zw}.

Let $h$ be the Planck constant and we are interested in the semiclassical limit that $h\to0$ in this paper. In a quantization procedure $a\to\Op_h(a)$, we assign a quantum observable $\Op_h(a)$ on $L^2(\R)$ to a classical observable $a\in C^\infty_0(\R^2)$. Then $a$ is called a symbol of $\Op_h(a)$.

Write $x=(q,p)\in\R^2$, in which $q$ and $p$ denote position and momentum variables respectively. Define the position and momentum self-adjoint operators $\hat q=\Op_h(q)$ and $\hat p=\Op_h(p)$:
$$\hat q\psi(q):=q\psi(q)\quad\text{and}\quad\hat p\psi:=\frac{h}{2\pi i}\frac{d\psi(q)}{dq}\quad\text{for }\psi\in C^\infty_0(\R).$$ 
So we have that
$$[\hat q,\hat p]:=\hat q\hat p-\hat p\hat q=\frac{ih}{2\pi}\Id.$$
Here, $\Id$ is the identity map that $\Id\psi=\psi$. 

The Weyl quantization of the Hamiltonian in \eqref{eq:Ham} is
$$\hat H=\Op_h(H)=\frac12\alpha\hat q^2+\frac12\beta\hat p^2+\frac\gamma2(\hat q\hat p+\hat p\hat q).$$
It generates the Schr\"odinger flow such that for a quantum state $\psi(0)\in L^2(\R)$,
$$\psi(0)\to\psi(t)=e^{-2\pi it\hat H/h}\psi(0).$$
So $\psi(t)$ solves the Schr\"odinger equation
$$\frac{ih}{2\pi}\frac{\partial\psi(t)}{\partial t}=\hat H\psi(t).$$
The quantization of the hyperbolic map $M$ on $\R^2$ is the Schr\"odinger flow at $t=1$:
\begin{equation}\label{eq:hatM}
\hat M=e^{-2\pi i\hat H/h}.
\end{equation}
Consider $v=(v_1,v_2)\in\R^2$. Define the phase space translation operator
$$\hat T_v:=\exp\left(-\frac{2\pi i}{h}(v_1\hat p-v_2\hat q)\right).$$
It readily follows that $\hat T_v^\star=\hat T_{-v}$. Moreover,
\begin{equation}\label{eq:MTM}
\hat M\hat T_v\hat M^{-1}=\hat T_{vM}.
\end{equation}
Notice also that for $u=(u_1,u_2)$ and $v=(v_1,v_2)$,
\begin{equation}\label{eq:TuTv}
\hat T_u\hat T_v=e^{\frac{2\pi i(u\wedge v)}{2h}}\hat T_{u+v},
\end{equation}
in which $u\wedge v$ is the symplectic product of $u$ and $v$ defined in \eqref{eq:sympprod}.

The function $\psi$ on $\R$ that defines a quantum state on $\T^1$ should be periodic in position and in momentum. That is, $\psi$ is invariant under the phase translations $\hat T_n$ for $n\in\Z^2$. In particular,
$$\hat T_{(1,0)}\psi=e^{2\pi i\kappa_1}\psi\quad\text{and}\quad\hat T_{(0,1)}\psi=e^{2\pi i\kappa_2}\psi.$$
Here, we allow the phase shifts $e^{2\pi i\kappa_1}$ and $e^{2\pi i\kappa_2}$ for some $\kappa=(\kappa_1,\kappa_2)\in\T^2$, because under such phase shifts the function defines the same quantum state. It then follows from such periodicity that
$$\hat T_{(1,0)}\hat T_{(0,1)}=\hat T_{(0,1)}\hat T_{(1,0)}$$
restricted to the quantum states on $\T^1$. But in the view of \eqref{eq:TuTv}, since $(1,0)\wedge(0,1)=-1$, it requires that $e^{2\pi i/h}=1$. Hence,
$$N:=\frac1h\in\N.$$
We always assume this condition throughout the paper. Under such condition, the space of quantum states $\HH_{N,\kappa}$ on $\T^2$ is an $N$-$\dim$ space that consists of distributions of the form
\begin{equation}\label{eq:quantumstates}
\psi(q)=\sum_{k\in\Z}\Psi(k)\delta(q-(k+\kappa_1)/N),\quad\text{in which }\Psi(k+N)=e^{-2\pi\kappa_2}\Psi(k).
\end{equation}
So $\HH_{N,\kappa}$ is a Hilbert space equipped with the inner product
$$\langle\psi,\phi\rangle=\frac{1}{N}\sum_{k=1}^N\Psi(k)\ol{\Phi(k)}.$$

\begin{defn}[Quantum cat maps]
Let $M$ be a classical cat map and $\hat M$ be defined in \eqref{eq:hatM}. Then for any $N\in\N$, there exists\footnote{See the detailed discussion in \cite[Section 6]{BDB}.} $\kappa\in\T^2$ such that $\hat M:\HH_{N,\kappa}\to\HH_{N,\kappa}$. We fix such choice of $\kappa$ (that depends on $M$ and $N$) and simply denote the Hilbert space of quantum states as $\HH_N$. There is then an eigenbasis $\{\phi_j\}_{j=1}^N\subset\HH_N$ such that $\hat M\phi_j=e^{i\theta_j}\phi_j$ for $0\le\theta_1\le\theta_2\le\cdots\le\theta_N<2\pi$.
\end{defn}

\begin{rmk}[Hecke eigenbases]
Assume in addition that $M=\Id\mod 4$. Then one can introduce the Hecke theory associated with $\hat M$. That is, there are a group of operators, called the Hecke operators, which commute with $\hat M$ acting on $\HH_N$. There is therefore a joint eigenbasis in $\HH_N$, i.e. Hecke eigenbasis, of all the Hecke operators and $\hat M$. The Hecke theory in cat maps was introduced by Kurlberg-Rudnick and we refer to \cite{KR1} for the precise construction.
\end{rmk}

Any phase translation $\hat T_v$ acts on $\HH_N$ only if $\hat T_v$ commutes with $\hat T_n$ for all $n\in\Z^2$. Applying \eqref{eq:TuTv} again, $e^{2\pi i(v\wedge n)/h}=1$ for all $n\in\Z^2$. So $v\in\Z^2/N$. For notational convenience, we write
$$\hat T_N(n):=\hat T_{n/N}.$$
Let $a\in C^\infty(\T^2)$ be a classical observable. Define its Weyl quantization as an operator on $\HH_N$:
\begin{equation}\label{eq:OpN}
\Op_N(a)=\sum_{n\in\Z^2}\tilde a(n)\hat T_N(n).
\end{equation}
Here, $\tilde a(n)$ is the Fourier coefficients of $a$ that
$$a(x)=\sum_{n\in\Z^2}\tilde a(n)e^{2\pi i(n\wedge x)}.$$
In the quantum cat system that $M$ is linear, we have the following exact Egorov's theorem. The proof is straightforward from its linear nature and we provide it here.
\begin{thm}[Egorov's theorem]\label{thm:Egorov}
Let $a\in C^\infty(\T^2)$. Then
$$\hat M^{-t}\circ\Op_N(a)\circ\hat M^t=\Op_N(a\circ M^t)\quad\text{for all }t\in\Z.$$
\end{thm}
\begin{proof}
It suffices to show the case when $\Op_N(a)=\hat T_N(n)$ and $t=1$. Observe that 
$$\hat T_N(n)=\Op_N\left(e^{2\pi i(n\wedge x)}\right).$$
By \eqref{eq:MTM}, we have that
\begin{eqnarray*}
\hat M^{-1}\hat T_N(n)\hat M&=&\hat T_{nM^{-1}/N}\\
&=&\Op_N\left(e^{2\pi i(nM^{-1}/N\wedge x)}\right)\\
&=&\Op_N\left(e^{2\pi i(n/N\wedge xM)}\right)\\
&=&\hat T_N(n)\circ\hat M.
\end{eqnarray*}
Here, we used the fact that the symplectic map $M$ preserves the symplectic product so $(nM^{-1}/N)\wedge x=(n/N)\wedge(xM)$.
\end{proof}

Notice that from \eqref{eq:quantumstates} (see \cite[Lemma 4]{KR1} for a short proof)
\begin{equation}\label{eq:Tntr}
\Tr\left(\hat T_N(n)\right)=\begin{cases}
N & \text{if }n=0\mod N,\\
0 & \text{otherwise}.
\end{cases}
\end{equation}
We then derive the following trace formula in $\HH_N$.
\begin{thm}[Trace formula]\label{thm:trace}
Let $a\in C^\infty(\T^2)$. Then
$$\Tr(\Op_N(a))=\sum_{j\in\Z^2}\tilde a(Nj).$$
\end{thm}

\section{Equidistribution and quantum ergodicity at small scales}\label{sec:SS}
Recall that $N=1/h\in\N$ as $h\to0$. Distribution of the eigenstates of $\{\phi_j\}_{j=1}^N$ is studied through
$$\langle\Op_N(a)\phi_j,\phi_j\rangle$$
for appropriate classical observables $a$.

\begin{itemize}
\item For equidistribution at small scale $r=r(N)$ in the physical space $\T^1$, we choose $a=\chi_{B_1(q,r)}$ for $B_1(q,r)\subset\T^1$ so
$$\langle\Op_N(\chi_{B_1(q,r)})\phi_j,\phi_j\rangle=\int_{B_1(q,r)}|\phi_j|^2\,d\Vol.$$
\item For equidistribution at small scale $r=r(N)$ in the phase space $\T^2$, i.e. small scale quantum ergodicity, we can only choose smooth functions $a\approx\chi_{B_2(x,r)}$ for $B_2(x,r)\subset\T^2$.
\end{itemize}

In both cases, notice that the quantization of $a$ in cat maps \eqref{eq:OpN} is via Fourier series of $a$. It is natural to approximate indicator functions of balls in $\T^d$, $d=1,2$, by trigonometric polynomials. 

Roughly speaking, to approximate $\chi_{B_1(q,r)}$ or $\chi_{B_2(x,r)}$, we need trigonometric polynomials of degree $D=D(r)$ such that $1/D=o(r)$. These trigonometric polynomials are the appropriate versions of Beurling-Selberg polynomials, which are well studied \cite{Har, Ho, HV}. Here we recall \cite[Lemma 2.5]{LR} that is explicit for our purpose.
\begin{lemma}\label{lemma:majmin}
Let $B_d(0,r)\subset\T^d$ and $D=D(r)$ such that $rD\ge1$. There exist trigonometric polynomials $a_r^\pm$ such that
\begin{enumerate}[(i).]
\item $a_r^-(y)\le\chi_{B_d(0,r)}(y)\le a_r^+(y)$ for all $y\in\T^d$,
\item $\wt{a_r^\pm}(n)=0$ if $|n|\ge D$,
\item $\wt{a_r^\pm}(0)=\Vol(B_d(0,r))+O(r^{d-1}/D)$,
\item $\left|\wt{a_r^\pm}(n)\right|\le cr^d$ for all $n\in\Z^2$, in which $c$ depends only on $d$.
\end{enumerate}
In particular, if $D=D(r)$ such that $1/D=o(r)$, then (iii) becomes
$$\wt{a_r^\pm}(0)=\Vol(B_d(0,r))+o\left(r^d\right).$$
\end{lemma}

Here, $a_r^+$ and $a_r^-$ above are called a majorant and a minorant of the indicator function $\chi_{B_d(0,r)}$. 
For any $x\in\T^d$, the trigonometric polynomials
\begin{equation}\label{eq:batx}
b^\pm_{x,r}(\cdot):=a^\pm(\cdot-x)
\end{equation}
majorize and minorize $\chi_{B_d(x,r)}$. Since $\wt{b_{x,r}^\pm}(n)=e^{-2\pi i(n\wedge x)}\wt{a^\pm}(n)$, $b^\pm_{x,r}$ satisfy the estimates in Lemma \ref{lemma:majmin} independent of $x\in\T^d$. In particular, with $D=D(r)$ such that $1/D=o(r)$,
\begin{equation}\label{eq:batxestimates}
\wt{b^\pm_{x,r}}(0)=\Vol(B_d(x,r))+o\left(r^d\right)\quad\text{and}\quad\left|\wt{b^\pm_{x,r}}(n)\right|\le cr^d.
\end{equation}

Define the $p$-moment
\begin{defn}[$p$-moment]
Let $1\le p<\infty$ and $a\in C^\infty(\T^2)$. Define the $p$-moment
$$V_p\left(N,\Op_N(a)\right):=\frac1N\sum_{j=1}^N\left|\left\langle\Op_N(a)\phi_j,\phi_j\right\rangle-\mu(a)\right|^p.$$
For notational simplicity, we also write $V_p(N,a)=V_p\left(N,\Op_N(a)\right)$.
\end{defn}

Inspired by \cite{LR}, we prove the following crucial lemma. From this lemma, equidistribution at small scales in the phase space (Theorems \ref{thm:SSQElog}, \ref{thm:SSQEpoly}, and \ref{thm:SSQEHecke}) is derived. 

\begin{lemma}\label{lemma:SNL2}
Let $L>0$ and $b^\pm_{x,r}$ be defined in \eqref{eq:batx} for $d=2$. Define
$$\SC^\pm(N,L):=\left\{1\le j\le N:\sup_{x\in\T^2}\left|\frac{\left\langle\Op_N(b^\pm_{x,r})\phi_j,\phi_j\right\rangle}{\mu\left(b^\pm_{x,r}\right)}-1\right|\ge L\right\}.$$
Denote $p'=p/(p-1)$. Assume that $D=D(r)$ such that $1/D=o(r)$ as $r\to0$. Then
$$\frac{\#\SC^\pm(N,L)}{N}\le\frac{cD^\frac{2p}{p'}}{L^p}\sum_{1\le|n|\le D}V_p\left(N,\hat T_N(n)\right),$$
in which $c$ depends on $p$. 
\end{lemma}

\begin{rmk}
In particular, when $p=1$, the above inequality reads
$$\frac{\#\SC^\pm(N,L)}{N}\le\frac{c}{L}\sum_{1\le|n|\le D}V_1\left(N,\hat T_N(n)\right).$$
It can be viewed as a variation in quantum cat maps of \cite[\S2.2]{LR} for toral eigenfunctions.
\end{rmk}

\begin{proof}
By the uniform control of the Fourier coefficients of $b^\pm_{x,r}$ in \eqref{eq:batxestimates}, 
\begin{eqnarray*}
&&\sup_{x\in\T^2}\left|\frac{\left\langle\Op_N(b^\pm_{x,r})\phi_j,\phi_j\right\rangle}{\mu\left(b^\pm_{x,r}\right)}-1\right|^p\\
&=&\sup_{x\in\T^2}\left|\frac{\sum_{n\in\Z^2}\wt{b_{x,r}^\pm}(n)\left\langle\hat T_N(n)\phi_j,\phi_j\right\rangle}{\wt{b_{x,r}^\pm}(0)}-1\right|^p\\
&=&\sup_{x\in\T^2}\left|\frac{\sum_{1\le|n|\le D}\wt{b_{x,r}^\pm}(n)\left\langle\hat T_N(n)\phi_j,\phi_j\right\rangle}{\wt{b_{x,r}^\pm}(0)}\right|^p\\
&\le&\sup_{x\in\T^2}\left(\sum_{1\le|n|\le D}\left|\frac{\wt{b_{x,r}^\pm}(n)}{\wt{b_{x,r}^\pm}(0)}\right|^{p'}\right)^\frac{p}{p'}\left(\sum_{1\le|n|\le D}\left|\left\langle\hat T_N(n)\phi_j,\phi_j\right\rangle\right|^p\right)\\
&\le&cD^\frac{2p}{p'}\sum_{1\le|n|\le D}\left|\left\langle\hat T_N(n)\phi_j,\phi_j\right\rangle\right|^p.
\end{eqnarray*}
Here, we used H\"older's inequality with exponents $p$ and $p'$. Hence, using Chebyshev's inequality,
\begin{eqnarray*}
\frac{\#\SC^\pm(N,L)}{N}&\le&\frac{1}{NL^p}\sum_{j=1}^N\sup_{x\in\T^2}\left|\frac{\left\langle\Op_N(b^\pm_{x,r})\phi_j,\phi_j\right\rangle}{\mu\left(b^\pm_{x,r}\right)}-1\right|^p\\
&\le&\frac{cD^\frac{2p}{p'}}{NL^p}\sum_{j=1}^N\sum_{1\le|n|\le D}\left|\left\langle\hat T_N(n)\phi_j,\phi_j\right\rangle\right|^p\\
&\le&\frac{cD^\frac{2p}{p'}}{L^p}\sum_{1\le|n|\le D}V_p\left(N,\hat T_N(n)\right).
\end{eqnarray*}
\end{proof}

Notice that when $a\in C^\infty(\T^1)$, i.e. $a$ depends only on the position variable $q$, 
$$a(q)=\sum_{m\in\Z}\tilde a(m)e^{2\pi imq}.$$
In this case, the quantization of $a$ on $\T^2$ is 
$$\Op_N(a)=\sum_{m\in\Z}\tilde a(m)\hat T_N(m,0).$$
We prove the following lemma, from which equidistribution at small scales in the physical space (Theorems \ref{thm:SSElog}, \ref{thm:SSQEpoly}, and \ref{thm:SSQEHecke}) is derived. The proof is similar as in Lemma \ref{lemma:SNL2} so we omit it here.

\begin{lemma}\label{lemma:SNL1}
Let $L>0$ and $b^\pm_{x,r}$ be defined in \eqref{eq:batx} for $d=1$. Then
$$\frac{\#\SC^\pm(N,L)}{N}\le\frac{cD^\frac{p}{p'}}{L^p}\sum_{1\le|m|\le D}V_p\left(N,\hat T_N(m,0)\right),$$
in which $c$ depends on $p$. 
\end{lemma}

\section{Logarithmical scales}\label{sec:SSlog}
To prove equidistribution at small scales using Lemma \ref{lemma:SNL2}, we need to estimate the $p$-moments of basic Fourier modes $\hat T_N(n)$. Denote the Ehrenfest time
$$T_E:=\frac{\log N}{\lambda}.$$
The following proposition provides the estimate of $2$-moments of $\hat T_N(n)$ by $T_E$. 

\begin{prop}\label{prop:V2TnN}
Let $\{\phi_j\}_{j=1}^N$ be an eigenbasis of a quantum cat map $\hat M$ in $\HH_N$. Suppose that 
$$1\le|n|<N\quad\text{and}\quad0<\delta<1-\frac{\log|n|}{\log N}.$$
Then
$$V_2\left(N,\hat T_N(n)\right)\le\frac{1}{\delta T_E}.$$
\end{prop}
\begin{proof}
Since $\hat M\phi_j=e^{i\theta_j}$, we have that
\begin{eqnarray*}
\left\langle\hat T_N(n)\phi_j,\phi_j\right\rangle&=&\left\langle\hat T_N(n)e^{-it\theta_j}\phi_j,e^{-it\theta_j}\phi_j\right\rangle\\
&=&\left\langle\hat T_N(n)\hat M^{-t}\phi_j,\hat M^{-t}\phi_j\right\rangle\\
&=&\left\langle\hat M^t\circ\hat T_N(n)\circ\hat M^{-t}\phi_j,\phi_j\right\rangle\\
&=&\left\langle\hat T_N\left(nM^t\right)\phi_j,\phi_j\right\rangle,
\end{eqnarray*}
by Egorov's theorem in Theorem \ref{thm:Egorov}. Then compute that
\begin{eqnarray}
&&V_2\left(N,\hat T_N(n)\right)\nonumber\\
&=&\frac1N\sum_{j=1}^N\left|\left\langle\hat T_N(n)\phi_j,\phi_j\right\rangle\right|^2\nonumber\\
&=&\frac1N\sum_{j=1}^N\left|\left\langle\frac1T\sum_{t=0}^{T-1}\hat M^t\circ\hat T_N(n)\circ\hat M^{-t}\phi_j,\phi_j\right\rangle\right|^2\nonumber\\
&=&\frac1N\sum_{j=1}^N\left|\left\langle\frac1T\sum_{t=0}^{T-1}\hat T_N\left(nM^t\right)\phi_j,\phi_j\right\rangle\right|^2\nonumber\\
&\le&\frac{1}{NT^2}\sum_{j=1}^N\left\langle\left(\sum_{t=0}^{T-1}\hat T_N\left(nM^t\right)\right)^\star\left(\sum_{s=0}^{T-1}\hat T_N\left(nM^s\right)\right)\phi_j,\phi_j\right\rangle\nonumber\quad\text{by Cauchy-Schwarz inequality}\\
&=&\frac{1}{NT^2}\sum_{t,s=0}^{T-1}\sum_{j=1}^N\left\langle\hat T_N\left(nM^t\right)^\star\hat T_N\left(nM^s\right)\phi_j,\phi_j\right\rangle\nonumber\\
&=&\frac{1}{NT^2}\sum_{t,s=0}^{T-1}\Tr\left(\hat T_N\left(nM^t\right)^\star\hat T_N\left(nM^s\right)\right).\label{eq:V2TnN}
\end{eqnarray}
For $1\le|n|<N$, since
$$0<\delta<1-\frac{\log|n|}{\log N},$$
there is $\delta_1$ such that
$$0<\delta<\delta_1<1-\frac{\log|n|}{\log N}.$$
We then have that $|n|<N^{1-\delta_1}$. Set 
$$T=\delta T_E=\frac{\delta\log N}{\lambda}.$$ 
Now if $0\le t,s\le T-1$, then 
$$|nM^t-nM^s|\le|nM^t|+|nM^s|\le2|e^{\lambda T}n|=2N^\delta|n|<2N^\delta N^{1-\delta_1}=2N^{\delta-\delta_1}<N.$$
It implies that for $0\le t,s\le T-1$, 
\begin{equation}\label{eq:MtMs}
nM^t=nM^s\mod N\quad\text{only if}\quad nM^t=nM^s,\text{ i.e. }t=s.
\end{equation}
Notice that from \eqref{eq:Tntr},
$$\Tr\left(T_N(k)^\star T_N(j)\right)=\begin{cases}
N & \text{if }j=k\mod N,\\
0 & \text{otherwise}.
\end{cases}$$
Therefore, in the view of \eqref{eq:MtMs}, we have that
$$\Tr\left(\hat T_N\left(nM^t\right)^\star\hat T_N\left(nM^s\right)\right)=\begin{cases}
N & \text{if }t=s,\\
0 & \text{otherwise}.
\end{cases}$$
Hence, \eqref{eq:V2TnN} continues as
$$V_2\left(N,\hat T_N(n)\right)\le\frac{1}{NT^2}\sum_{t,s=1}^{T-1}\Tr\left(\hat T_N\left(nM^t\right)^\star\hat T_N\left(nM^s\right)\right)=\frac1T=\frac{1}{\delta T_E}.$$
\end{proof}

Now we prove equidistribution at logarithmical scales in the physical space.
\begin{proof}[Proof of Theorem \ref{thm:SSElog}]
For $0\le\alpha<1/2$, let $D=(\log N)^\beta$ with some $\beta\in(\alpha,1/2)$. Then $1/D=o(r)$ as $r\to0$ since $r=(\log N)^{-\alpha}$. By Proposition \ref{prop:V2TnN}, we can choose 
$$\frac12<\delta<1-\frac{\log D}{\log N}$$ 
such that
$$V_2\left(N,\hat T_N(m,0)\right)\le\frac{1}{\delta T_E}\quad\text{for all }1\le|m|\le D.$$
Applying Lemma \ref{lemma:SNL1} for $d=1$ and $p=2$, we immediately have that
\begin{eqnarray*}
\frac{\#\SC^\pm(N,L)}{N}&\le&\frac{cD^\frac{p}{p'}}{L^2}\sum_{1\le|m|\le D}V_2\left(N,\hat T_N(m,0)\right)\\
&\le&\frac{CD}{L^2}\sum_{1\le|m|\le D}\frac{1}{\delta T_E}\\
&\le&\frac{CD^2}{L^2\log N}.
\end{eqnarray*}
Since $0<\beta<1/2$, let
$$\gamma=\frac{1-2\beta}{3}>0\quad\text{and}\quad L=\frac{1}{(\log N)^\gamma}.$$ 
Then
$$\frac{CD^2}{L^2\log N}=\frac{C}{(\log N)^{1-2\beta-2\gamma}}=\frac{C}{(\log N)^\gamma}\to0\quad\text{as }N\to\infty.$$
Denote
$$S(N)=\{1,...,N\}\setminus\{\SC^+(N,L)\cup\SC^-(N,L)\}.$$
It is evident that $S(N)$ has full density in $\{1,...,N\}$ as $N\to\infty$. 

If $j\not\in\SC^+(N,L)$, then
$$\sup_{x\in\T^2}\left|\frac{\left\langle\Op_N(b^+_{q,r})\phi_j,\phi_j\right\rangle}{\mu\left(b^+_{q,r}\right)}-1\right|\le L=\frac{1}{(\log N)^\gamma}.$$
This means that
$$\left\langle\Op_N(b^+_{q,r})\phi_j,\phi_j\right\rangle\le(1+L)\mu\left(b^+_{q,r}\right)\quad\text{if }j\not\in\SC(N,L),$$
uniformly for all $q\in\T$. In the view of \eqref{eq:batxestimates},
$$\int_{B_1(q,r)}|\phi_j|^2\,d\Vol\le\Vol(B_1(q,r))+o(r)\quad\text{if }j\not\in\SC^+(N,L).$$
A similar analysis implies that the above inequality holds for $j\not\in\SC^-(N,L)$ with inequality reversed. Hence,
$$\int_{B_1(q,r)}|\phi_j|^2\,d\Vol=\Vol(B_1(q,r))+o(r)\quad\text{if }j\in S(N),$$
uniformly for all $q\in\T^1$. 
\end{proof}

We then prove equidistribution at logarithmical scales in the phase space, i.e. small scale quantum ergodicity.
\begin{proof}[Proof of Theorem \ref{thm:SSQElog}]
For $0\le\alpha<1/4$, let $D=(\log N)^\beta$ with some $\beta\in(\alpha,1/4)$. Then $1/D=o(r)$ as $r\to0$ since $r=(\log N)^{-\alpha}$. By Proposition \ref{prop:V2TnN}, we can choose 
$$\frac12<\delta<1-\frac{\log D}{\log N}$$ 
such that
$$V_2\left(N,\hat T_N(n)\right)\le\frac{1}{\delta T_E}\quad\text{for all }1\le|n|\le D.$$
Applying Lemma \ref{lemma:SNL2} for $d=2$ and $p=2$, we immediately have that
\begin{eqnarray*}
\frac{\#\SC^\pm(N,L)}{N}&\le&\frac{cD^\frac{2p}{p'}}{L^2}\sum_{1\le|n|\le D}V_2\left(N,\hat T_N(n)\right)\\
&\le&\frac{CD^2}{L^2}\sum_{1\le|n|\le D}\frac{1}{\delta T_E}\\
&\le&\frac{CD^4}{L^2\log N}.
\end{eqnarray*}
Since $0<\beta<1/4$, let
$$\gamma=\frac{1-4\beta}{3}>0\quad\text{and}\quad L=\frac{1}{(\log N)^\gamma}.$$ 
Then
$$\frac{CD^4}{L^2\log N}=\frac{C}{(\log N)^{1-4\beta-2\gamma}}=\frac{C}{(\log N)^\gamma}\to0\quad\text{as }N\to\infty.$$
Denote
$$S(N)=\{1,...,N\}\setminus\{\SC^+(N,L)\cup\SC^-(N,L)\}.$$
It is evident that $S(N)$ has full density in $\{1,...,N\}$ as $N\to\infty$. If $j\in S(N)$, we deduce that
$$\sup_{x\in\T^2}\left|\frac{\left\langle\Op_N(b^\pm_{x,r})\phi_j,\phi_j\right\rangle}{\mu\left(b^+_{x,r}\right)}-1\right|\le L=\frac{1}{(\log N)^\gamma}.$$
In the view of \eqref{eq:batxestimates}, this means that
$$\lim_{N\to\infty}\left\langle\Op_N(b^\pm_{x,r})\phi_j,\phi_j\right\rangle=\Vol(B_2(x,r))+o(r^2)\quad\text{if }j\in S(N),$$
uniformly for all $x\in\T^2$.
\end{proof}

\section{Polynomial scales}\label{sec:SSpoly}
\subsection{Polynomial scales for full density integers}
To prove Theorem \ref{thm:SSQEpoly}, we need the results in Kurlberg-Rudnick \cite[Proposition 8 and Theorem 17]{KR2} which provide the control of the $4$-moment of $\hat T_N(n)$. These results are used to prove QUE in \cite[Theorems 1 and 2]{KR2} for a full density subset of integers; they were improved by Bourgain \cite[Theorem 3]{Bo} to include a larger set of integers (still full density). But the improvement does not provide smaller scale in quantum ergodicity so we use \cite{KR2} here.

\begin{thm}\label{thm:V4TnN}
Let $\{\phi_j\}_{j=1}^N$ be an eigenbasis of a quantum cat map $\hat M$ in $\HH_N$. There are $\delta>0$ and a full density subset $S(\N)$ of integers such that for all $\ve>0$ and $N\in S(\N)$ we have that
$$V_4\left(N,\hat T_N(n)\right)\le\frac{C|n|^{8+\ve}}{Ne^{(\log N)^\delta}}\quad\text{for }|n|>0,$$
in which $C$ depends only on $M$ and $\delta$. In particular, if $N\in S(\N)$ here, then
$$P(N)\gg\sqrt Ne^{(\log N)^\delta}.$$
\end{thm}

W now prove equidistribution in the physical space and quantum ergodicity at small scales in Theorem \ref{thm:SSQEpoly}.

\begin{proof}[Proof of Theorem \ref{thm:SSQEpoly}]
We first use Lemma \ref{lemma:SNL2} for $d=2$ and $p=4$ to prove quantum ergodicity at small scales.  
\begin{eqnarray*}
\frac{\#\SC^\pm(N,L)}{N}&\le&\frac{cD^\frac{2p}{p'}}{L^4}\sum_{1\le|n|\le D}V_4\left(N,\hat T_N(n)\right)\\
&\le&\frac{CD^6}{L^4}\sum_{1\le|n|\le D}\frac{|n|^{8+\ve}}{Ne^{(\log N)^\delta}}\\
&\le&\frac{CD^6}{L^4Ne^{(\log N)^\delta}}\sum_{1\le|n|\le D}|n|^{8+\ve}\\
&\le&\frac{CD^{16+\ve}}{L^4Ne^{(\log N)^\delta}}.
\end{eqnarray*}
For $0\le\alpha<1/16$, let $D=N^\beta$ for some $\beta\in(\alpha,1/16)$. Then $1/D=o(r)$ since $r=N^{-\alpha}$. Choose $\ve>0$ small enough such that $1-(16+\ve)\beta>0$. Let
$$\gamma=\frac{1-(16+\ve)\beta}{4}>0\quad\text{and}\quad L=\frac{1}{N^\gamma}.$$ 
Compute that
$$\frac{CD^{16+\ve}}{L^4Ne^{(\log N)^\delta}}=\frac{CN^{(16+\ve)\beta}}{N^{-4\gamma}Ne^{(\log N)^\delta}}=\frac{C}{N^{1-(16+\ve)\beta-4\gamma}e^{(\log N)^\delta}}=\frac{C}{e^{(\log N)^\delta}}\to0$$
as $N\to\infty$. It follows that $\SC^\pm(N,L)$ both have zero density in $\{1,...,N\}$ as $N\to\infty$. Denote
$$S(N)=\{1,...,N\}\setminus\{\SC^+(N,L)\cup\SC^-(N,L)\}.$$
It is evident that $S(N)$ has full density in $\{1,...,N\}$ as $N\to\infty$. If $j\in S(N)$, we deduce that
$$\sup_{x\in\T^2}\left|\frac{\left\langle\Op_N(b^\pm_{x,r})\phi_j,\phi_j\right\rangle}{\int_{\T^2} b^\pm_{x,r}}-1\right|\le L=\frac{1}{N^\gamma}.$$
In the view of \eqref{eq:batxestimates}, this means that
$$\lim_{N\to\infty}\left\langle\Op_N(b^\pm_{x,r})\phi_j,\phi_j\right\rangle=\Vol(B_2(x,r))+o(r^2)\quad\text{if }j\in S(N),$$
uniformly for all $x\in\T^2$.

We then use Lemma \ref{lemma:SNL1} for $d=1$ and $p=4$ to prove equidistribution at small scales in the physical space.
\begin{eqnarray*}
\frac{\#\SC^\pm(N,L)}{N}&\le&\frac{cD^\frac{p}{p'}}{L^4}\sum_{1\le|m|\le D}V_4\left(N,\hat T_N(m,0)\right)\\
&\le&\frac{CD^3}{L^4}\sum_{1\le|m|\le D}\frac{|m|^{8+\ve}}{Ne^{(\log N)^\delta}}\\
&\le&\frac{CD^3}{L^4Ne^{(\log N)^\delta}}\sum_{1\le|m|\le D}|m|^{8+\ve}\\
&\le&\frac{CD^{12+\ve}}{L^4Ne^{(\log N)^\delta}}.
\end{eqnarray*}
For $0\le\alpha<1/12$, let $D=N^\beta$ for some $\beta\in(\alpha,1/12)$. Then $1/D=o(r)$ since $r=N^{-\alpha}$. Choose $\ve>0$ small enough such that $1-(12+\ve)\beta>0$. Let
$$\gamma=\frac{1-(12+\ve)\beta}{4}>0\quad\text{and}\quad L=\frac{1}{N^\gamma}.$$ 
Compute that
$$\frac{CD^{12+\ve}}{L^4Ne^{(\log N)^\delta}}=\frac{CN^{(12+\ve)\alpha}}{N^{-4\gamma}Ne^{(\log N)^\delta}}=\frac{C}{N^{1-(12+\ve)\beta-4\gamma}e^{(\log N)^\delta}}=\frac{C}{e^{(\log N)^\delta}}\to0$$
as $N\to\infty$. It follows that $\SC^\pm(N,L)$ both have zero density in $\{1,...,N\}$ as $N\to\infty$. Denote
$$S(N)=\{1,...,N\}\setminus\{\SC^+(N,L)\cup\SC^-(N,L)\}.$$
It is evident that $S(N)$ has full density in $\{1,...,N\}$ as $N\to\infty$. If $j\in S(N)$, we deduce that
$$\sup_{x\in\T^2}\left|\frac{\left\langle\Op_N(b^\pm_{x,r})\phi_j,\phi_j\right\rangle}{\int_{\T^2} b^\pm_{x,r}}-1\right|\le L=\frac{1}{N^\gamma}.$$
In the view of \eqref{eq:batxestimates}, similarly as in the proof of Theorem \ref{thm:SSElog}, this means that
$$\int_{B_1(q,r)}|\phi_j|^2\,d\Vol=\Vol(B_1(q,r))+o(r)\quad\text{if }j\in S(N),$$
uniformly for all $q\in\T^1$.
\end{proof}

\subsection{Polynomial scales for Hecke eigenbasis} 
To prove Theorem \ref{thm:SSQEHecke}, we need an estimate of the $4$-moments for Hecke eigenbasis \cite[Theorem 10]{KR1}.

\begin{thm}\label{thm:V4TnNHecke}
Let $\{\phi_j\}_{j=1}^N$ be a Hecke eigenbasis of a quantum cat map $\hat M$ in $\HH_N$. Then for all $\ve>0$ we have that
$$V_4\left(N,\hat T_N(n)\right)\le\frac{C|n|^{16}}{N^{2-\ve}}\quad\text{for }|n|>0,$$
in which $C$ depends only on $M$ and $\ve$.
\end{thm}

\begin{proof}[Proof of Theorem \ref{thm:SSQEHecke}]
We first use Lemma \ref{lemma:SNL2} for $d=2$ and $p=4$ to prove quantum ergodicity at small scales.  
\begin{eqnarray*}
\frac{\#\SC^\pm(N,L)}{N}&\le&\frac{cD^\frac{2p}{p'}}{L^4}\sum_{1\le|n|\le D}V_4\left(N,\hat T_N(n)\right)\\
&\le&\frac{CD^6}{L^4}\sum_{1\le|n|\le D}\frac{|n|^{16}}{N^{2-\ve}}\\
&\le&\frac{CD^6}{L^4N^{2-\ve}}\sum_{1\le|n|\le D}|n|^{16}\\
&\le&\frac{CD^{24}}{L^4N^{2-\ve}}.
\end{eqnarray*}
Similar argument as in the proof of Theorem \ref{thm:SSQEpoly} shows quantum ergodicity at scales $N^{-\alpha}$, $0\le\alpha<1/12$, for Hecke eigenbasis.

We then use Lemma \ref{lemma:SNL1} for $d=1$ and $p=4$ to prove quantum ergodicity at small scales.  
\begin{eqnarray*}
\frac{\#\SC^\pm(N,L)}{N}&\le&\frac{cD^\frac{p}{p'}}{L^4}\sum_{1\le|m|\le D}V_4\left(N,\hat T_N(m,0)\right)\\
&\le&\frac{CD^3}{L^4}\sum_{1\le|m|\le D}\frac{|m|^{16}}{N^{2-\ve}}\\
&\le&\frac{CD^3}{L^4N^{2-\ve}}\sum_{1\le|m|\le D}|m|^{16}\\
&\le&\frac{CD^{20}}{L^4N^{2-\ve}}.
\end{eqnarray*}
Similar argument as in the proof of Theorem \ref{thm:SSQEpoly} shows equidistribution in the physical space at scales $N^{-\alpha}$, $0\le\alpha<1/10$, for Hecke eigenbasis.
\end{proof}

\section*{Acknowledgements}
I benefited from discussions with Fr\'ed\'eric Faure, Andrew Hassell, Ze\'ev Rudnick, Stephane Nonnenmacher, and Steve Zelditch. This work began when I participated the Special Year Program in Analysis at Australian National University in 2018. I thank the institute for the hospitality.

\end{document}